\documentclass[journal]{IEEEtran}

\usepackage{psfrag} 
\usepackage{cuted}
\usepackage{float}
\usepackage{graphicx} 
\usepackage{amsmath}  
\usepackage{amsmath} 
\usepackage{amsthm}
\usepackage{flushend}
\usepackage{array}
\usepackage{bbm}
\usepackage{mathrsfs}
\usepackage{epstopdf}

\DeclareMathAlphabet{\mathpzc}{OT1}{pzc}{m}{it}

\IEEEoverridecommandlockouts                                                         
\IEEEoverridecommandlockouts                              

\usepackage{graphicx}
\usepackage[fancythm,fancybb,morse,ieee]{jphmacros2e}
\usepackage{amsmath}

\ifCLASSINFOpdf
 
\else

\fi

\begin{document}

\title{Deterministic and Stochastic Fixed-time Stability of Discrete-time Autonomous Systems}

\author{Farzaneh Tatari, and Hamidreza Modares,~\IEEEmembership{Senior~Member,~IEEE}
\thanks{F. Tatari, and H. Modares are with Michigan State University, East Lansing, MI 48824 USA (e-mail: tatarifa@msu.edu, modaresh@msu.edu).}

}


\maketitle

\begin{abstract}
This paper studies deterministic and stochastic fixed-time stability of autonomous  nonlinear discrete-time (DT) systems. Lyapunov conditions are first presented under which the fixed-time stability of deterministic DT system is certified. Extensions to systems under deterministic perturbations as well as stochastic noise  are then considered. For the former, the sensitivity to perturbations for fixed-time stable DT systems is analyzed, and it is shown that fixed-time attractiveness is resulted from the presented Lyapunov conditions. For the latter, sufficient Lyapunov conditions for fixed-time stability in probability of nonlinear stochastic DT systems are presented. The fixed upper bound of the settling-time function is derived for both fixed-time stable and fixed-time attractive systems, and the stochastic settling-time function fixed upper bound is derived for stochastic DT systems. Illustrative examples are given along with simulation results to verify the introduced results. 
\end{abstract}


\IEEEpeerreviewmaketitle

\section{Introduction}

\setlength\parindent{7pt} The Lyapunov stability theory has a longstanding history as a powerful tool in control theory to obtain many important results in the design of a variety of controllers and adaptation laws. The basic framework of the Lyapunov stability theory provides conditions under which their satisfaction guarantees the stability of the system in some sense. While finding a function satisfying these conditions, called Lyapunov function, is generally challenging, controllers and update laws can be developed to make a candidate Lyapunov function enforce the stability conditions. 
	
The Lyapunov theory generally provides conditions to assure the states of a system convergence to an equilibrium state. The qualitative guarantees that are provided for the convergence time determine the stability type, ranging from asymptotic stability, exponential stability, finite-time stability to fixed-time stability. While asymptotic stability and exponential stability provide assurance that the system’s states eventually converge to an equilibrium, many real-world practical systems demand intense time response constraints, which makes these types of stabilities insufficient. Therefore, a surge of interest has emerged in the control community in studying finite-time stability to design control systems and adaptation laws that exhibit finite-time convergence to an equilibrium point. 
	
	Finite-time stability \cite{bhat2000finite} has been studied for continuous-time (CT) and discrete-time (DT) deterministic and stochastic systems \cite{chen2010finite,yin2011finite,rajpurohit2017stochastic,haddad2020finite}.
	Moreover, finite-time stability concept has been extensively applied for the finite-time control of DT \cite{li2013discrete,sun2017discrete,zhao2014neural} and CT \cite{liu2018neural,liu2019direct,li2019adaptive} systems, as well as finite-time identification \cite{adetola2008finite,zhao2019performance,yang2018adaptive,wang2019robust,wang2019finite,vahidi2020memory,lehrer2010parameter,tatari2021finite,tatari2021nonlinear,tatari2022finite1}. In the finite-time stability, however, the settling (i.e., convergence) time, depends on the system's initial condition, and, thus, cannot be specified a priori. Moreover, when the magnitude  of the initial condition is large, it can lead to an unacceptable convergence time guarantee. Fixed-time stability, on the other hand, imposes a stronger requirement on the settling time and provides convergence guarantees with a pre-specified bound on the settling time, independent of the initial condition. Fixed-time stability of deterministic and stochastic CT systems, respectively, studied in \cite{polyakov2011nonlinear} and \cite{yu2019fixed}, have been widely studied within the frameworks of fixed-time control design \cite{zhang2021observer,liu2020fixed,liu2020team,garg2020prescribed,dimanidis2020output, min2022fixed,ren2020fixed,zhao2021stochastic}, fixed-time observer design \cite{basin2017hypersonic,gao2020general,zhang2019singular,zhang2020stabilization,ni2017fixed,yu2017design} and fixed-time identification \cite{noack2016fixed,zhu2020online,wang2020fixed,efimov2020fixed,rios2017time,tatari2021fixed}.
	
	While most real-world systems are CT in nature, DT systems are of great importance since systems are typically discretized and controlled with digital computers  and  micro-controllers in real-world applications. Even though finite-time stability of DT deterministic \cite{tatari2021finite,tatari2021nonlinear,hamrah2019discrete,haddad2020finite} and stochastic \cite{haddad2021lyapunov,lee2021finite} systems are recently studied, fixed-time stability of DT deterministic and stochastic systems is surprisingly unsettled, despite its practical importance.  This gap motivates us to present fixed-time Lyapunov stability conditions that pave the way for the realization of fixed-time control and identification of DT systems through designing appropriate controllers and adaptation laws, respectively. 
	
Lyapunov theory can also be leveraged to study the behavior of uncertain systems. There are typically two types of uncertainties in control systems: randomness which is caused by a noise in a stochastic system, and  deterministic unknown perturbations with known bounds (here, we call the deterministic systems affected by deterministic perturbations as perturbed deterministic systems). The stability results are typically presented in terms of stability in probability for stochastic systems' stability \cite{kushner1967stochastic,yu2019fixed,yin2011finite,chen2010finite,lee2021finite}, which guarantees convergence in probability to an equilibrium  point, and in terms of attractiveness to a bounded set for perturbed systems. 
	
	In this paper, we develop fixed-time stability conditions for both deterministic and stochastic DT nonlinear systems. First, fixed-time stability for equilibria of deterministic DT autonomous systems is defined. That is, a settling-time function is defined with a fixed upper bound independent of the initial condition. We then present Lyapunov theorems for fixed-time stability of both unperturbed and perturbed deterministic DT systems. Moreover, the sensitivity of fixed-time stability properties to perturbations of systems is investigated under the assumption of the existence of a locally Lipschitz discrete Lyapunov function. 
	It is ensured that fixed-time stability is preserved under perturbations in the form of fixed-time attractiveness. Furthermore, sufficient Lyapunov conditions for fixed-time stability in probability of stochastic DT systems and their stochastic settling-time function are presented. 
	The presented framework will pave the way for designing control laws with guaranteed satisfaction of a given performance measure in fixed time. Moreover, the presented stability results can be leveraged to develop fixed-time observers and identifiers for deterministic and stochastic DT systems, 
	which are of great importance in control of safety-critical systems that highly rely on a system model and a state estimator to make less-conservative and feasible decisions. This is because fixed-time stability allows the system to preview and quantify probable errors in state estimators and identifiers considerably fast, which can be employed by the control system to avoid conservatism.

	This paper is organized as follows. Section 2 describes
the fixed-time stability of deterministic DT systems. The sensitivity to deterministic perturbation for fixed-time stable DT systems is analyzed in Section 3. Section 4 explains the fixed-time stability in probability of stochastic DT systems. Section 5 represents the verification of the introduced method through illustrative examples along with simulation results. \vspace{6pt}

	\textbf{Notations:} In this paper, the following notations are employed. $\mathbb{R}$, $\mathbb{R}^+$, $\mathbb{Z}$,  $\mathbb{N}^+$, and ${\mathbb{N}}$ represent, respectively, the set of real numbers, non-negative real numbers, integer numbers, natural numbers except zero, and natural numbers. Moreover, $\mathbb{R}^n$ represents the set of $n \times 1$ real column vectors. $\Arrowvert . \Arrowvert$ is used to denote induced 2-norm for matrices and the Euclidean norm for vectors. The trace of a matrix $A$ is indicated with $tr(A)$. $|.|$ denotes the absolute value of any scalar $x$. $\lfloor . \rfloor: \mathbb{R} \mapsto \mathbb{Z} $ is the floor function. $\Delta(.)$ is the DT difference operator for deterministic systems and is defined for a function $V(y(k)):\mathbb{R}^n \mapsto \mathbb{R^+}$ as $\Delta V(y(k+1))= V(y(k+1)) - V(y(k))$.
	
	All random variables are assumed to be defined on a probability space $(\Omega,\mathcal{F},\mathbb{P})$, with $\Omega$ as the sample space, $\mathcal{F}$ as its associated Borel $\sigma$-algebra and $\mathbb{P}$ as the probability measure. For a random variable $\nu: \Omega \longrightarrow \mathbb{R}^n$ defined on the probability space $(\Omega,\mathcal{F},\mathbb{P})$, with some abuse of notation, the statement $\nu \in \mathbb{R}^n$ is used to state the dimension of the random variable. $\mathbb{E}[X]$ denotes the expected value of the random variable $X$ on the probability space $(\Omega,\mathcal{F},\mathbb{P})$. It is assumed that the probability space $(\Omega, \mathcal{F}, \mathbb{P})$ admits a sequence of mutually independent identically distributed random vectors $\nu(k), k \in \mathbb{N}$.


 \section{Fixed-time Stability for Deterministic Discrete-time Systems} 

	In this section, the fixed-time stability of autonomous unperturbed deterministic DT systems is defined and the Lyapunov theorem specifying the sufficient conditions for their fixed-time stability is presented.
	
	Consider the following nonlinear DT system, 
	\begin{align} 
		y(k+1)= F(y(k)) , \label{syst} 
	\end{align}
	where $F:\mathcal{D}_y \mapsto \mathcal{D}_y, F(0)=0$ is a nonlinear function on $\mathcal{D}_y$, and $\mathcal{D}_y$ is an open set with $0 \in \mathcal{D}_y$. Moreover, $y(k) \in \mathcal{D}_y \subseteq \mathbb{R}^n, \, k \in {{ {\mathbb{N}}}}$ is the system state vector. For an initial condition $y(0)$, define the solution sequence $y(k), \, k \in  {\mathbb{N}}_{y(0)} \subseteq {{ {\mathbb{N}}}}$, where $\mathbb{N}_{y(0)}$ is the maximal interval of existence of $y(k)$ after which the solution may cease outside the domain of $F(.)$. Then, the solution sequence $y(k), \, k \in {\mathbb{N}}_{y(0)} \subseteq {{ {\mathbb{N}}}}$ is uniquely defined in forward time for every initial condition $y(0) \in \mathcal{D}_y$ irrespective of whether or not the function $F(.)$ is a continuous function \cite{haddad2020finite}.

	Before proceeding, the following definitions are needed.
	\vspace{6 pt}
	\begin{definition} {\textit{(Locally Lipschitz function)}}
	 A function $f(x)$ is locally Lipschitz on a domain $\Omega \subset \mathbb{R}^n$ if for each point in $\Omega$ there exist a neighborhood $\Omega_0$ and a positive constant $L$ such that
	 \begin{align}
	 || f(x)-f(y)|| \le L \,\, ||x-y||, \forall x \in \Omega_0, y \in \Omega_0. \label{Lips}
	 \end{align}
	 Moreover, $L$ is called the Lipschitz constant of $f(x)$.
	 \end{definition}

	The following definition extends the fixed-time stability definition presented in \cite{polyakov2011nonlinear} for CT systems to DT systems. 

	\vspace{6 pt}
	\begin{definition} \label{fixedtime} {\textit{(Fixed-time stability)}}
			Consider the DT nonlinear system \eqref{syst}. The zero solution of $y(k) \equiv 0$ to the system \eqref{syst} is said to be fixed-time stable, if there exist an open neighborhood $\mathcal{N}_y \subseteq \mathcal{D}_y$ of the origin and a settling time function $K:\mathcal{N}_y \backslash \{0\} \mapsto \mathbb{N}^+ $, such that:
		\begin{enumerate} 
		\item The system \eqref{syst} is Lyapunov stable. That is, for every $\epsilon > 0$, there exists a $\delta > 0$ such that if  $||y(0)|| \le \delta$, then  $||y(k)|| \le \epsilon$ for all $k \in \{0,...,K(y(0))-1\}$.
		\item For every initial condition $y(0) \in \mathcal{N}_y \backslash \{0\}$, the solution sequence $y(k)$ of \eqref{syst} reaches the equilibrium point and remains there after $k > K(y(0))$ and $\forall y(0) \in \mathcal{N}_y$, where $K:\mathcal{N}_y \backslash \{0\} \mapsto \mathbb{N}^+ $. 
		\item  The settling-time function $K(y(0))$ is bounded, i.e., $\exists K_{max} \in \mathbb{N}^+ : K(y(0)) \le K_{max}, \forall y(0) \in {\mathcal {N}}_y \backslash \{0\}$.
	\end{enumerate}	
	DT nonlinear system \eqref{syst} is globally fixed-time stable if it is fixed-time stable with $\mathcal{N}_y=\mathcal{D}_y=\mathbb{R}^n$.
	\end{definition} \vspace{6pt}
	
	\begin{remark}
		If only conditions 1) and 2) of the above definitions are satisfied, the finite-time stability \cite{bhat2000finite} is resulted. In contrast, the fixed-time stability imposes the additional condition 3). This requirement makes the upper bound of the settling time in the fixed-time stability independent of the initial condition, in contrast to the finite-time stability. Therefore, the fixed-time stability is a stronger type of stability than the finite-time stability.
	\end{remark} \vspace{6pt}
	
	The following theorem provides sufficient conditions under which the system \eqref{syst} is fixed-time stable.  \vspace{6pt}

	\begin{theorem} 
		Consider the nonlinear DT system \eqref{syst}. Suppose there is a Lyapunov function $V:\mathcal{D}_y \mapsto \mathbb{R}^+$ where $\mathcal{D}_y$ is an open neighborhood around the origin and there exist a neighborhood $\Omega_y \subset \mathcal{D}_y$ of the origin such that 
		\begin{align}
			&V(y(0))=0, \label{init}\\&
			V(y(k))>0, \,\,\, y(k) \in \Omega_y \backslash \{0\}, \label{VG}\\&
			\Delta V(y(k+1)) \le -\alpha \min \{ \frac{V(y(k))}{\alpha},\nonumber\\& \max\{ V^{r_1}(y(k)) , V^ {r_2}(y(k))\}\},\,\,\,y(k) \in \Omega_y \backslash \{0\}, \label{Lyap}
		\end{align}
		for some positive constants $0 < \alpha <1$, $0<r_1<1 $, and $r_2>1$.
		Then, the system \eqref{syst} is fixed-time stable and has a settling time function $K: \mathcal{N}_y \mapsto \mathbb{N}^+ $ that satisfies
		\begin{align}
			&K(y(0)) \le \lfloor {\alpha^{\frac{1}{1-r_1}}(1-\alpha^{\frac{1}{1-r_1}})}\rfloor + \lfloor \alpha^{-1}(\alpha^{\frac{1}{1-r_2}} -1)\rfloor +3,\label{FTB1}
		\end{align} 
		for all $y(0) \in \mathcal{N}_y \backslash \{0\}$ where $\mathcal{N}_y$ is an open neighborhood of the origin.
		Moreover, if $\mathcal{D}_y=\mathbb{R}^n$, $V(.)$ is radially unbounded and \eqref{Lyap} holds on $\mathbb{R}^n$, then system \eqref{syst} is globally fixed-time stable.
	\end{theorem}
		\textbf{Proof}
		The Lyapunov stability of the system \eqref{syst} can be concluded using similar arguments as of \cite{haddad2020finite} (see Theorem 4.1). The proof of fixed-time stability consists of three parts. In the first part, we show that for $V(y(0)) > {\alpha}^{\frac{1}{1-r_2}}$, the settling time function is $K(y(0))=1$.
		In the second part, we show that if ${\alpha}^{\frac{1}{1-r_1}}<V(y(0)) < {\alpha}^{\frac{1}{1-r_2}}$, there exists a settling-time function with a fixed upper bound $K^{*}$ (i.e., $K(y(0)) \le K^{*}$) such that one has $ V(y(k)) =0, \,\,\, \forall k > K^{*}$. Finally in the third part, for $V(y(0)) \le {\alpha}^{\frac{1}{1-r_1}}$, the Lyapunov function reaches $V(k)=0$ with settling-time function $K(y(0))=1$. 
	
		Since $0<r_1<1 $ and $r_2>1$, one has
		\begin{align}
			 V^{r_2}(y(k)) <  V^{r_1}(y(k)) ,\,\,\,\,\,\, \forall V(y(k)) \le 1, \label{app1}
		\end{align} 
		and
		\begin{align}
			V^{r_1}(y(k))  \le   V^{r_2}(y(k)) ,\,\,\,\,\,\, \forall V(y(k)) > 1. \label{app2}
		\end{align}  
		We, first, prove part 1 where $V(y(0)) > {\alpha}^{\frac{1}{1-r_2}}$. In this case, since $ {\alpha}^{\frac{1}{1-r_2}}>1$, using \eqref{app2}, \eqref{Lyap} leads to 
		\begin{align}
			&\Delta V(y(k+1)) \le -\alpha \min \{ \frac{V(y(k))}{\alpha},  V^ {r_2}(y(k))\}.
		\end{align}
		Moreover, since $V(y(0)) > {\alpha}^{\frac{1}{1-r_2}}$, the above inequality for $k=0$ yields 
		\begin{align}
			&\Delta V(y(1)) \le -V(y(0)). \label{ini} 
		\end{align}
		Now, \eqref{ini} implies that the settling time function is $K(y(0))=1$, for $V(y(0)) > {\alpha}^{\frac{1}{1-r_2}}$. 
	
		For part 2 where ${\alpha}^{\frac{1}{1-r_1}}<V(y(k)) < {\alpha}^{\frac{1}{1-r_2}}$, based on \eqref{Lyap}, first we show that $V(k)$ reduces to $V(y(k)) \le 1$ after some time where this time is upper bounded by a fixed constant $K_1^*$. 
		
		Note that for $1<V(y(k)) < {\alpha}^{\frac{1}{1-r_2}}$, using \eqref{app2}, one has 
		\begin{align}
			&\min \{ \frac{V(y(k))}{\alpha},\max\{ V^{r_1}(y(k)) , V^ {r_2}(y(k))\}\} =\nonumber \\&
		\min \{ \frac{V(y(k))}{\alpha},  V^ {r_2}(y(k))\}= V^ {r_2}(y(k)), \label{11}
		\end{align}
		Then, \eqref{11} and \eqref{Lyap}, lead to
		\begin{align}
			&V(y(k+1)) \le V(y(k)) -\alpha V^{r_2}(y(k)) . \label{diff1}
		\end{align}
		The condition \eqref{diff1} holds for $k=0,...,{K_1^*-1}$ where $1<V(y(k)) < {\alpha}^{\frac{1}{1-r_2}}$. Therefore, using \eqref{diff1} for $k=0,1,...,K_1^*-1$, one has
		\begin{align}
			V(y(1))-V(y(0)) & \le  -\alpha V^{r_2}(y(0)),\nonumber \\
			V(y(2))-V(y(1)) & \le  -\alpha V^{r_2}(y(1)),\nonumber \\
			&\vdots \nonumber \\
			V(y(K_1^*-1))-V(y(K_1^*-2)) &\le  -\alpha V^{r_2}(y(K_1^*-2)),\nonumber \\
			V(y(K_1^*))-V(y(K_1^*-1)) &\le  -\alpha V^{r_2}(y(K_1^*-1)),\nonumber
		\end{align}
		which leads to 
		\begin{align}
			V(y(K_1^*))-V(y(0)) \le  \sum_{k=0}^{K^*_1-1} { -\alpha V^{r_2}(y(k))}. \label{vk1}
		\end{align}
		Since $V(y(k))<V(y(k-1))$, \eqref{vk1} can be rewritten as 
		\begin{align}
			V(y(K_1^*-1))-V(y(0)) \le -{K_1^*} \alpha V^{r_2}(y(K_1^*-1)), \label{vk11}
		\end{align}
		that leads to
		\begin{align}
			{K_1^*} \le \frac{V(y(0))-V(y(K_1^*-1))}{\alpha V^{r_2}(y(K_1^*-1))}.\label{2k1}
		\end{align}
		Using $1<V(y(0)) < {\alpha}^{\frac{1}{1-r_2}}$ for $k<K_1^*$, \eqref{2k1} implies
		\begin{align}
			{K_1^*} \le \frac{{\alpha}^{\frac{1}{1-r_2}}-1}{\alpha},
		\end{align}
		which leads to the integer upper bound for $K_1^*$ as follows 
		\begin{align}
			{K_1^*} \le \lfloor \alpha^{-1}(\alpha^{\frac{1}{1-r_2}} -1)\rfloor + 1. \label{boundk1}
		\end{align}
		
		Note that since for $k > K_1^*$ one has $V(y(k))\le 1$. Thus, for ${\alpha}^{\frac{1}{1-r_1}}<V(y(k)) \le 1$, using \eqref{app1} one has
		\begin{align}
				&\min \{ \frac{V(y(k))}{\alpha},\max\{ V^{r_1}(y(k)) , V^ {r_2}(y(k))\}\} =\nonumber \\&
		\min \{ \frac{V(y(k))}{\alpha},  V^ {r_1}(y(k))\}= V^ {r_1}(y(k)) , \label{18}
		\end{align}
		and \eqref{18} and \eqref{Lyap} result in
		\begin{align}
			V(y(k+1)) \le V(y(k))-\alpha V^{r_1}(y(k)). \label{diff2}
		\end{align}
		Using \eqref{diff2}, there exists a time  $k>K_2^*$ such that $V(k)$ reaches $V(y(k))\le \alpha^{\frac{1}{1-r_1}}$ where $K_2^*$ is a fixed positive integer. Using \eqref{diff2} for $k=K_1^*,K_1^*+1,...,K_2^*-1$ one obtains
		\begin{align}
			V(y(K_1^*+1))-V(y(K_1^*)) & \le  -\alpha V^{r_1}(y(K_1^*)),\nonumber \\
			V(y(K_1^*+2))-V(y(K_1^*+1)) & \le  -\alpha V^{r_1}(y(K_1^*+1)),\nonumber \\
			&\vdots \nonumber \\
			V(y(K_2^*-1))-V(y(K_2^*-2)) &\le  -\alpha V^{r_1}(y(K_2^*-2)),\nonumber \\
			V(y(K_2^*))-V(y(K_2^*-1)) &\le  -\alpha V^{r_1}(y(K_2^*-1)), \label{k2ineq}
		\end{align}
		
which leads to 
		\begin{align}
			K^*_2-K^*_1 \le \frac{V(K_1^*) - V(K_2^*-1)}{\alpha V^{r_1}(y(K^*_2-1))}. \label{K1K2final}
		\end{align}
Since, $\alpha^{\frac{1}{1-r_1}}<V(y(k))<1$ for $k=K_1^*,K_1^*+1,...,K_2^*-1$, \eqref{K1K2final} reduces to

		\begin{align}
			K^*_2  \le  K^*_1 + \lfloor \alpha^{\frac{1}{1-r_1}} ({1 - \alpha^{\frac{1}{1-r_1}}})\rfloor +1. \label{bound}
		\end{align}
		Using \eqref{boundk1}, \eqref{bound} is rewritten as follows
		\begin{align}
			K^*_2 \le   \lfloor \alpha^{-1}(\alpha^{\frac{1}{1-r_2}} -1)\rfloor + \lfloor \alpha^{\frac{1}{1-r_1}} ({1 - \alpha^{\frac{1}{1-r_1}}})\rfloor +2. \label{bound11}
		\end{align}
		
		At time $k>K_2^*$ for which $V(y(k))\le \alpha^{\frac{1}{1-r_1}}$, \eqref{Lyap} reduces to 
		\begin{align}
			\Delta V(y(k+1)) \le -{V(y(k))}, \label{P3}
		\end{align}
		which leads to $V(y(k+1))=0$ for $k \ge K_2^*+1$. This completes the proof of part 2.
		
		The proof of part 3 where $V(y(0)) \le {\alpha}^{\frac{1}{1-r_1}}$ is also derived based on \eqref{P3} where $V(k)$ reaches zero with $K(y(0))=1$.
		
		Hence, the Lyapunov function reaches $V(y(k))=0$ with the settling-time function $K(y(0))$ such that

\begin{align}
	& K(y(0))=1,\,\,\,\, \nonumber \\&V(y(0))>{\alpha}^{\frac{1}{1-r_2}}\,\,\, and\,\,\, V(y(0))<{\alpha}^{\frac{1}{1-r_1}},  \label{bdB1}
\end{align}
and
\begin{align}
&K(y(0)) \le \lfloor \alpha^{-1}(\alpha^{\frac{1}{1-r_2}} -1)\rfloor +\lfloor \alpha^{\frac{1}{1-r_1}} ({1 - \alpha^{\frac{1}{1-r_1}}})\rfloor+3,\nonumber \\&\alpha^{\frac{1}{1-r_1}}<V(y(0)) \le {\alpha}^{\frac{1}{1-r_2}}.  \label{bdB}
\end{align}
 Therefore, the system is fixed-time stable, and the system trajectory converges to the origin with the settling-time function given in \eqref{FTB1}. This completes the proof.

Moreover, if $\mathcal{N}_y=\mathcal{D}_y=\mathbb{R}^n$ and $V(.)$ is radially unbounded, the global fixed-time stability follows using the same procedure.				\hfill {$\square$}

\section{Sensitivity to Deterministic Perturbation for Fixed-time Stable Discrete-time Systems} \label{Sec:2}
	The system \eqref{syst} usually describes a nominal model of the system that works under ideal conditions. Nevertheless, many real-world systems are under uncertainties and disturbances that affect the system's behavior. To account for these uncertainties, a more accurate representation of the system can be given by the following deterministic perturbed model 
	\begin{align}
		y(k+1)= F(y(k)) + g(k,y(k)), \label{pert}
	\end{align}
	where $g$ represents perturbation caused by disturbances, uncertainties, or modeling errors. This section investigates the solution behavior of the deterministic perturbed system \eqref{pert} in a neighborhood of the fixed-time stable equilibrium of the nominal system \eqref{syst}.
	\begin{assumption}
	The perturbation term $g$ is bounded, i.e.,  
		\begin{align}
			\mathop {\sup }\limits_{\mathbb{N}^+ \times \mathcal{D}_y} \|g(k,y(k))\| < \delta_0, \label{assump2}
		\end{align}
		for some $\delta_0 < \infty$.
	\end{assumption} \vspace{6pt}
	
		The following definition extends the fixed-time attractiveness definition presented in \cite{polyakov2011nonlinear} for CT systems to DT systems.

	\begin{definition} {(Fixed-time attractiveness)}
	The perturbed system \eqref{pert} is said to be fixed-time attractive by a bounded set $\mathcal{Y}$ around the equilibrium point, if $\forall y(0) \in {\mathcal {N}}_y$ the solution sequence $y(k)$ of \eqref{pert} reaches $\mathcal{Y}$ in finite time $k>K(y(0))$ and remains there  for all $k> K(y(0))$, where $K:\mathcal{N}_y \backslash \{0\} \mapsto \mathbb{N}^+ $ is the settling-time function and the settling-time function $K(y(0))$ is bounded, i.e., $\exists K_{max} \in \mathbb{N}^+ : K(y(0)) \le K_{max}, \forall y(0) \in {\mathcal {N}}_y$.
		\end{definition}
	
	\vspace{6pt}
	The following lemma is required in the proof of Lyapunov-based fixed-time attractiveness of perturbed deterministic systems. \vspace{6pt}
	
	\begin{lemma}
	Let $V(y(k)):\mathcal{D}_y \mapsto \mathbb{R}^+$ be a fixed-time Lyapunov function for the the nominal (unperturbed) system \eqref{syst}, i.e., $V(y(k))$ satisfies conditions \eqref{init}-\eqref{Lyap} for the system  \eqref{pert} when  $g=0$. Let also $V(y(k))$ be locally Lipschitz continuous on $\mathcal{D}_y$ with Lipschitz constant $L_V$ and Assumption 1 hold. Then, for the perturbed deterministic system \eqref{pert}, $V(k)$ satisfies
		\begin{align}
		\Delta V(y(k+1)) \le & -\alpha \min \{ \frac{V(y(k))}{\alpha},\nonumber\\& \max\{ V^{r_1}(y(k)) , V^ {r_2}(y(k))\}\} \nonumber \\ & + L_V \|g(k,y(k))\|, \label{modlyap}
	\end{align}
	where $\Delta V(y(k+1))$ is computed along the solution of the unperturbed deterministic system.
	\end{lemma}
	\textbf{Proof}
	The proof is similar to \cite{scokaert1997discrete}, which is developed for exponential stability, and is thus omitted.  \hfill {$\square$}   
 \vspace{6pt}

	The following theorem provides the behavior of deterministic fixed-time stable DT systems under bounded deterministic perturbations.  
	
	\vspace{6pt}
	\begin{theorem} 
		Suppose there exists a Lyapunov function $V:\Omega_y \mapsto \mathbb{R}^+$ which is locally Lipschitz on an open neighborhood $\Omega_y$ of the origin with Lipschitz constant $L_V$ and satisfies \eqref{init}-\eqref{Lyap} for the nominal system \eqref{syst} for some  real positive numbers $\alpha, r_1, r_2>0$ such that $0<\alpha<1 $, $0<r_1<1 $, and $r_2>1$. Let Assumption 1 hold. Then, around the origin, the system \eqref{pert} is fixed-time attractive to the following bound
		\begin{align}
			b_{y}=\{y \in \Omega_y : V(y) \le \mathcal{B} \}, \label{bd}
		\end{align} 
		where
	\begin{align}
	\mathcal{B}=\left \{  \begin{tabular}{ll} $(\frac{m_1 L_V \delta_0}{\alpha})^{\frac{1}{r_2}},\,\,\,\,\,\,\,\,\,\,\,\,\,\,\,\,1<V(y(0))<{\alpha}^{\frac{1}{1-r_2}}$, \\ $(\frac{m_2 L_V \delta_0}{\alpha})^{\frac{1}{r_1}},\,\,\,\,\,\,\,\,\,\,\,\,\,\,\,\,{\alpha}^{\frac{1}{1-r_1}}<V(y(0)) \le 1$, \end{tabular}  \right .  \label{bdB2}
\end{align} 
and its fixed-time bounded settling-time function is $K(y(0)) \le K^*$ where

	\begin{align}
		K^*=\left \{  \begin{tabular}{ll} $\lfloor \alpha_c^{-1}(\alpha^{\frac{1}{1-r_2}} -1)\rfloor+1,\,\,\,\,\,\,\,1<V(y(0))<{\alpha}^{\frac{1}{1-r_2}}$, \\ $\lfloor \alpha_d^{-1} (\alpha^{\frac{r_1}{r_1-1}} - \alpha)\rfloor+1,\,\,\,{\alpha}^{\frac{1}{1-r_1}}<V(y(0)) \le 1$, \end{tabular}  \right .  \label{fixt}
	\end{align}

		$\alpha_c=(1-\frac{1}{m_1})\alpha$, $\alpha_d=(1-\frac{1}{m_2})\alpha$. The constants $m_1>1$ and $m_2>1$ are selected such that 

		\begin{align}
		\left \{  \begin{tabular}{ll} $\alpha \mathcal{B}^{r_2} - m_1 L_V \delta_0 > 0,\,\,\,\,\,\,\,\,\,\,\,\,\,\,\,\,V(0)>1$, \\ $\alpha \mathcal{B}^{r_1} - m_2 L_V \delta_0 > 0,\,\,\,\,\,\,\,\,\,\,\,\,\,\,\,\,V(0) \le 1$
		\end{tabular}  \right.  \label{m1m2}
	\end{align}

	\end{theorem}
\vspace{6 pt}	
	
\textbf{Proof} 
		According to Theorem 1, the origin is the fixed-time stable equilibrium for the unperturbed or nominal system \eqref{syst}. 
		
		Lemma 1 and \eqref{assump2} imply that 
		\begin{align}
			\Delta V(y&(k+1)) \le -\alpha \min \{ \frac{V(y(k))}{\alpha},\nonumber\\& \max\{ V^{r_1}(y(k)) , V^ {r_2}(y(k))\}\} + L_V \delta_0. \label{pertV}
		\end{align}

		For $1<V(y(0))<{\alpha}^{\frac{1}{1-r_2}}$, \eqref{pertV} leads to
		\begin{align}
			\Delta V(y(k+1)) \le -\alpha V^{r_2}(y(k)) + L_V \delta_0. \label{V2pert}
		\end{align}
		Having $1<V(y(0))<{\alpha}^{\frac{1}{1-r_2}}$ and $V(y(0))>\mathcal{B}$, and using  \eqref{m1m2} and $m_1>1$, one has
		\begin{align}
			\alpha \mathcal{B}^ {r_2} - m_1L_V \delta_0 >0 &\Rightarrow  -\alpha \mathcal{B}^ {r_2}  + m_1 L_V \delta_0 < 0, \nonumber \\ &\Rightarrow -\alpha \mathcal{B}^ {r_2}  +  L_V \delta_0 < 0, 
		\end{align}
		which results in
		\begin{align}
			 L_V \delta_0< \frac{1}{m_1}\alpha \mathcal{B}^ {r_2} . \label{delb}
		\end{align}
		For $y(0) \notin b_y$ ($V(y(0))>\mathcal{B}$) and $1<V(y(0))<{\alpha}^{\frac{1}{1-r_2}}$, \eqref{V2pert} and \eqref{delb} imply that   
		\begin{align}
			\Delta V(y(k+1)) \le  -\alpha V^{r_2}(k)+ \frac{1}{m_1}\alpha \mathcal{B}^{r_2}.\label{lab} 
		\end{align}
		Using $V(y(k))> \mathcal{B}$, \eqref{lab} is upper bounded as follows 
		\begin{align}
			\Delta V(y(k+1)) \le - \alpha_c {V^{r_2}(y(k))},\label{V1T2}
		\end{align} such that $\alpha_c=(1-\frac{1}{m_1})\alpha$ is positive.
		Using the results of part 2 in Theorem 1 proof, \eqref{V1T2} implies that for $y(0) \notin b_y$ and $1<V(y(0))<{\alpha}^{\frac{1}{1-r_2}}$ with $\alpha<m_1 L_V \delta_0$, $y(k)$ reaches the invariant set \eqref{bd} within the fixed time steps $K^*=\lfloor \alpha_c^{-1}(\alpha^{\frac{1}{1-r_2}} -1)\rfloor+1$ and remains there after. 
		
		Using \eqref{pertV}, for ${\alpha}^{\frac{1}{1-r_1}}<V(y(0)) \le 1$, one has 
		\begin{align}
			\Delta V(y(k+1)) \le -\alpha V^{r_1}(y(k)) + L_V \delta_0.\label{delV2} 
		\end{align}

		Having ${\alpha}^{\frac{1}{1-r_1}}<V(y(0)) \le 1$ and $V(y(0))>\mathcal{B}$, and using \eqref{m1m2} and $m_2>1$, one has
		\begin{align}
			\alpha \mathcal{B}^ {r_1} - m_2L_V \delta_0 >0 & \Rightarrow   -\alpha \mathcal{B}^ {r_1}  + m_2 L_V \delta_0 < 0, \nonumber \\ & \Rightarrow -\alpha \mathcal{B}^ {r_1}  +  L_V \delta_0 < 0 . \label{delBB}
		\end{align}
		From \eqref{delBB}, one obtains 
		\begin{align}
			  L_V \delta_0< \frac{1}{m_2}\alpha \mathcal{B}^ {r_1} . \label{delB2}
		\end{align}
		For $y(0) \notin b_y$ ($V(y(0))>\mathcal{B}$) and ${\alpha}^{\frac{1}{1-r_1}}<V(y(0)) \le 1$, then \eqref{delV2} and \eqref{delB2} imply that   
		\begin{align}
			\Delta V(y(k+1)) \le  -\alpha V^{r_1}(y(k))+ \frac{1}{m_2}\alpha \mathcal{B}^{r_1}. \label{finalV}
		\end{align}
		Using $V(y(k))> \mathcal{B}$, \eqref{finalV} is upper bounded as follows
		\begin{align}
			\Delta V(y(k+1)) \le - \alpha_d {V^{r_1}(y(k))},\label{finalV2}
		\end{align} such that $\alpha_d=(1-\frac{1}{m_2})\alpha$ is positive.
		Using the results of part 2 in Theorem 1 proof, \eqref{finalV2} implies that for $y(0) \notin b_y$ and ${\alpha}^{\frac{1}{1-r_1}}<V(y(0))<1$ with $m_2 L_V \delta_0<\alpha$, $y(k)$ reaches the invariant set \eqref{bd} within the fixed time steps $K^*=\lfloor \alpha_d^{-1} (\alpha^{\frac{r_1}{r_1-1}} - \alpha)\rfloor+1$ and remains in $b_y$ ever after.
		This completes the proof.\hfill {$\square$}

	\begin{remark}
		In \eqref{bd}, the bound $\mathcal{B}$ is either a function of $m_1$ or $m_2$, as given is \eqref{bdB2}. Notice that the fixed-time attractive bound \eqref{bdB2} increases by choosing large values for $m_1$ or $m_2$ and accordingly the fixed-time of convergence given in \eqref{fixt} decreases. Therefore, the bigger we choose the bounded set $\mathcal{B}$, the shorter the fixed-time of convergence and vice-versa.   
	\end{remark}
	

\section{Fixed-time Stability in Probability for Stochastic Discrete-time Systems}

Consider the DT nonlinear stochastic system given by
\begin{align}
\mathpzc{y}(k+1)=&\mathpzc{f}(\mathpzc{y}(k))+\mathpzc{g}(\mathpzc{y}(k)) \nu(k) \triangleq F(\mathpzc{y}(k), \nu(k)),\nonumber\\ &
\mathpzc{y}(0) \stackrel{\text { a.s. }}{=} \mathpzc{y}_{0}, \quad k \in \mathbb{N}, \label{stochsys}
\end{align}
where, for every $k \in \mathbb{N}, \mathpzc{y}(k) \in \mathcal{D} \subseteq \mathbb{R}^{n}$ is a $\mathcal{D}$-valued stochastic process with $\mathpzc{y}_{0} \in \mathcal{D}$, and $\nu(k)\in \mathbb{R}^n , k \in \mathbb{N}$, is the independent and identically distributed zero-mean stochastic process on $(\Omega, \mathcal{F}, \mathbb{P})$. $\mathpzc{f}$ : $\mathcal{D} \rightarrow \mathcal{D}$ and $\mathpzc{g}: \mathcal{D} \rightarrow \mathbb{R}^{n \times n}$ are continuous functions with $\mathpzc{f}(0)=0$ and $\mathpzc{g}(0)=0$ where $\mathpzc{y}_e=0$ is the equilibrium of the system \eqref{stochsys}, if and only if $\mathpzc{y}(.) $ is $\mathbb{P}$-almost surely (a.s.) equal to zero (i.e., $\mathpzc{y}(.) \stackrel{\text { a.s. }}{=} 0$) and is a solution of \eqref{stochsys}.

A stochastic process $\mathpzc{y}:[0, k] \times \Omega \rightarrow \mathcal{D}$ is a solution sequence of \eqref{stochsys} on the discrete-time interval $[0, \kappa]$ with initial condition $\mathpzc{y}(0) \stackrel{\text { a.s. }}{=} \mathpzc{y}_{0}$ if $\mathpzc{y}(k)$ satisfies \eqref{stochsys} almost surely. 



The following definitions are given for stability in probability for the zero solution $\mathpzc{y}(k) \stackrel{\text { a.s. }}{\equiv} 0$ of the DT nonlinear stochastic system \eqref{stochsys}. 

\begin{definition}{\cite{lee2021finite,qin2019lyapunov}}
\begin{itemize}
\item The zero solution $\mathpzc{y}(k) \stackrel{\text { a.s. }}{\equiv} 0$ to \eqref{stochsys} is Lyapunov stable in probability, if for every $\varepsilon>0$ and $\rho \in(0,1)$, there exist $\delta=\delta(\varepsilon, \rho)>0$ such that, for all $||\mathpzc{y}_{0}|| <\delta$,
$$
\mathbb{P}\left(\sup _{k \in \mathbb{N}}\|\mathpzc{y}(k)\|>\varepsilon\right) \leq \rho .
$$
\item The zero solution $\mathpzc{y}(k) \stackrel{\text { a.s. }}{\equiv} 0$ to \eqref{stochsys} is asymptotically stable in probability if it is Lyapunov stable in probability and, for every $\rho \in(0,1)$, there exists $\delta=\delta(\rho)>$ 0 such that if $||\mathpzc{y}_{0}|| <\delta$, then
$$
\mathbb{P}\left(\lim _{k \rightarrow \infty}\|\mathpzc{y}(k)\|=0\right) \geq 1-\rho .
$$
\item The zero solution $\mathpzc{y}(k) \stackrel{\text { a.s. }}{\equiv} 0$ to \eqref{stochsys} is globally asymptotically stable in probability if it is Lyapunov stable in probability and, for all $\mathpzc{y}_{0} \in \mathbb{R}^{n}$,
$$
\mathbb{P}\left(\lim _{k \rightarrow \infty}\|\mathpzc{y}(k)\|=0\right)=1 .
$$
\item The zero solution $\mathpzc{y}(k) \stackrel{\text { a.s. }}{\equiv} 0$ to \eqref{stochsys} is exponentially stable in probability if for some $0<\gamma<1$ independent of $\nu$, it is Lyapunov stable in probability and, for every $\rho \in(0,1)$, there exists $\delta=\delta(\rho)>$ 0 such that if $||\mathpzc{y}_{0}|| <\delta$, then
$$
\mathbb{P}\left(\lim _{k \rightarrow \infty}\| \gamma^k \mathpzc{y}(k)\|=0\right) \geq 1-\rho .
$$
\item The zero solution $\mathpzc{y}(k) \stackrel{\text { a.s. }}{\equiv} 0$ to \eqref{stochsys} is globally exponentially stable in probability if for some $0<\gamma<1$ independent of $\nu$, it is Lyapunov stable in probability and, for all $\mathpzc{y}_{0} \in \mathbb{R}^{n}$,
$$
\mathbb{P}\left(\lim _{k \rightarrow \infty}\| \gamma^k \mathpzc{y}(k)\|=0\right)=1 .
$$
\end{itemize}
\end{definition}

\begin{definition}{\cite{lee2021finite}}
 For the DT stochastic dynamical system \eqref{stochsys} and $V: \mathcal{D} \rightarrow \mathbb{R}^+$, the difference operator $\boldsymbol{\Delta} V$ of $\mathpzc{y}$ is given as follows,
$$
\boldsymbol{\Delta} V(\mathpzc{y}) = \mathbb{E}[V(F(\mathpzc{y}, \nu))]-V(\mathpzc{y}), \quad \mathpzc{y} \in \mathcal{D} .
$$
\end{definition}
Note that the difference operator in Definition $5$ is a deterministic function and does not involve the expectation of the system state trajectory and only involves the expectation over the random noise variable $\nu$. Moreover, the random vectors $\nu(k), k \in \mathbb{N}$, all have the same distribution. 

In the following, sufficient conditions for Lyapunov, asymptotic and exponential stability in probability for the system \eqref{stochsys} are given.

\begin{lemma} {\cite{haddad2021lyapunov,qin2019lyapunov}}: Consider the discrete-time nonlinear stochastic system \eqref{stochsys} and assume that there exists a continuous function $V: \mathcal{D} \rightarrow \mathbb{R}^+$ such that
$$
\begin{gathered}
V(0)=0, \\
V(\mathpzc{y})>0, \quad \mathpzc{y} \in \mathcal{D}, \quad \mathpzc{y} \neq 0, \\
\boldsymbol{\Delta} V(\mathpzc{y}) \leq 0, \quad \mathpzc{y} \in \mathcal{D}.
\end{gathered}
$$
Then the zero solution $\mathpzc{y}(k) \stackrel{\text { a.s. }}{\equiv} 0$ to \eqref{stochsys} is Lyapunov stable in probability. Moreover, if
$$
\boldsymbol{\Delta} V(\mathpzc{y})<0, \quad \mathpzc{y} \in \mathcal{D}, \quad \mathpzc{y} \neq 0,
$$
then the zero solution $\mathpzc{y}(k) \stackrel{\text { a.s. }}{\equiv} 0$ to \eqref{stochsys} is asymptotically stable in probability. 
Furthermore, if 
$$
\boldsymbol{\Delta} V(\mathpzc{y})< -\gamma V(\mathpzc{y}) , \quad 0<\gamma<1 , \quad \mathpzc{y} \in \mathcal{D}, \quad \mathpzc{y} \neq 0,
$$
then the zero solution $\mathpzc{y}(k) \stackrel{\text { a.s. }}{\equiv} 0$ to \eqref{stochsys} is exponentially stable in probability.
If $\mathcal{D}=\mathbb{R}^{n}$ and $V(\cdot)$ is radially unbounded, then the zero solution $\mathpzc{y}(k) \stackrel{\text { a.s. }}{\equiv} 0$ to \eqref{stochsys} is globally asymptotically or exponentially stable in probability under the defined Lyapunov conditions.
\end{lemma}

The following definition provides the characteristics of stochastic DT systems under which they are fixed-time stable in probability.

\vspace{6 pt}
	\begin{definition} \label{fixedtime} {\textit{(Fixed-time stability in probability)}}
			Consider the stochastic DT nonlinear system \eqref{stochsys}. The zero solution of $\mathpzc{y}(k) \stackrel{\text { a.s. }}{\equiv} 0$ to the system \eqref{stochsys} is said to be fixed-time stable in probability, if there exist a stochastic process called stochastic settling time function $K(\mathpzc{y}, \cdot)$, such that:
		\begin{enumerate} 
		\item The system \eqref{stochsys} is Lyapunov stable in probability. That is, for every $\epsilon > 0$ and $\rho \in (0,1)$, there exists a $\delta= \delta(\epsilon,\rho) > 0$ such that for all $\mathpzc{y}(0) \stackrel{\text { a.s. }}{=} \mathpzc{y}_{0} \in \mathcal{D} \backslash\{0\}$, if  $||\mathpzc{y(0)}|| \le \delta$, then  
$$
\mathbb{P}\left(\sup _{k \in\left[0, K\left(\mathpzc{y}_{0}, \nu\right)\right)}\left\|\mathpzc{y}(k)\right\|>\varepsilon\right) \leq \rho .
$$		
		\item For every initial condition $\mathpzc{y}(0) \stackrel{\text { a.s. }}{=} \mathpzc{y}_{0} \in \mathcal{D} \backslash\{0\}$, the solution sequence $\mathpzc{y}(k)$ is defined on $\left[0, K\left(\mathpzc{y}_{0}, \nu\right)\right)$, $\nu \in \Omega, \mathpzc{y}(k) \in \mathcal{D} \backslash\{0\}, k \in\left[0, K\left(\mathpzc{y}_{0}, \nu\right)\right), \nu \in \Omega$, and
$$
\mathbb{P}\left(\left\|\mathpzc{y}\left(K\left(\mathpzc{y}_{0}, \nu\right)\right)\right\|=0\right)=1 .
$$
		\item  The stochastic settling-time function $K(\mathpzc{y}, \cdot)$, for all $\mathpzc{y} \in \mathcal{D}$, is finite almost surely and there exist a fixed-time upper bound for the stochastic settling-time $K(\mathpzc{y}, \cdot)$, i.e., $\mathbb{E}[K(\mathpzc{y}_0, \nu)] \le K_{max}$ where $K_{max}$ is a positive integer.
	\end{enumerate}	
The zero solution $\mathpzc{y}(k) \stackrel{\text { a.s. }}{\equiv} 0$ to \eqref{stochsys} is globally fixed-time stable in probability if it is fixed time stable in probability with $\mathcal{D}=\mathbb{R}^{n}$.
	\end{definition} \vspace{6pt}

\begin{lemma}
Consider the nonlinear stochastic DT system \eqref{stochsys} and the scalar system 
\begin{align}
			&V(x(k+1)) = \gamma(V(x(k))), \,\,\,x(k) \in \mathbb{R}^n\label{gama},
		\end{align}
		where
\begin{align}
			&\gamma(V(x(k)))=V(x(k)) -\alpha \min \{ \frac{V(x(k))}{\alpha},\nonumber\\& \max\{ V^{r_1}(x(k)) , V^ {r_2}(x(k))\}\},\label{17}
		\end{align}
		such that $0 < \alpha <1$, $0<r_1<1 $, and $r_2>1$.
If there exists a continuous positive-definite function $V: \mathbb{R}^{n} \rightarrow {\mathbb{R}}^{+}$and the nondecreasing function $\gamma:{\mathbb{R}}^{+} \rightarrow {\mathbb{R}}^{+}$such that
$$
\mathbb{E}\left[V(F(\mathpzc{y}, \nu)] \leq \gamma(V(\mathpzc{y})), \quad y \in \mathbb{R}^{n},\right.
$$
then
$$
V\left(\mathpzc{y}_{0}\right) \leq x_{0}, \quad x_{0} \in {\mathbb{R}}^{+}
$$
implies
$$
\mathbb{E}[V(\mathpzc{y}(k))] \leq x(k), \quad k \in \mathbb{N},
$$
where the sequence $x(k), k \in \mathbb{N}$, satisfies \eqref{gama}.
\end{lemma}

\textbf{Proof}. This Lemma is an extension of finite-time stability conditions {\cite{lee2021finite}}, which is provided for fixed-time stability conditions. The proof is similar and is omitted. 	\hfill {$\square$}

The following theorem represents the sufficient Lyapunov conditions for fixed-time stability in probability for stochastic DT nonlinear systems.  

\begin{theorem}
 Consider the nonlinear stochastic system \eqref{stochsys}. If there exists a continuous and radially unbounded function $V: \mathbb{R}^{n} \rightarrow \mathbb{R}^+$ such that
\begin{align}
V(0) &=0 ,\label{22}\\
V(\mathpzc{y}) &>0, \quad \mathpzc{y} \in \mathbb{R}^{n} \backslash\{0\}, \label{23} \\
\mathbb{E}[V(F(\mathpzc{y}, \nu))] & \leq \gamma(V(\mathpzc{y})), \quad \mathpzc{y} \in \mathbb{R}^{n} \backslash\{0\}, \label{stochFxT}
\end{align}
where $\gamma(.)$ is given in \eqref{17}, then the zero solution $\mathpzc{y}(k) \stackrel{\text { a.s. }}{\equiv} 0$ to \eqref{stochsys} is globally fixed-time stable in probability. Moreover, there exists a stochastic settling-time $K: \mathbb{R}^{n} \rightarrow \mathbb{N}$ such that
\begin{align}
\mathbb{E}\left[K\left(\mathpzc{y}_{0}\right)\right] \leq \hat{K}\left(x_{0}\right) < K_{max},
\end{align}
where $K(\cdot)$ is almost surely finite stochastic settling-time function and $\hat{K}\left(x_{0}\right)$ is the finite settling-time function of \eqref{gama} and $K_{max}$ is the fixed upper bound for $\hat{K}\left(x_{0}\right)$ and $\mathbb{E}\left[K\left(\mathpzc{y}_{0}\right)\right]$.
\end{theorem}

\textbf{Proof} 
Based on \eqref{17} and \eqref{stochFxT}, one has
$$
\begin{aligned}
\mathbb{E}[V(F(\mathpzc{y}, \nu))]-V(\mathpzc{y}) & \leq \gamma(V(\mathpzc{y}))-V(\mathpzc{y}) \\
&<0, \quad \mathpzc{y} \in \mathbb{R}^{n} \backslash\{0\},
\end{aligned}
$$
and hence, it follows from Lemma 2 that the zero solution $\mathpzc{y}(k) \stackrel{\text { a.s. }}{=} 0$ to \eqref{stochsys} is globally asymptotically stable in probability.
Now, consider the nonlinear DT system \eqref{gama} and note that, by Theorem 1, the zero solution $x(k) \equiv$ 0 to \eqref{gama} is globally fixed-time stable and there exists $\hat{K}\left(x_{0}\right)<\lfloor {\alpha^{\frac{1}{1-r_1}}(1-\alpha^{\frac{1}{1-r_1}})}\rfloor + \lfloor \alpha^{-1}(\alpha^{\frac{1}{1-r_2}} -1)\rfloor +3$ such that
$$
x(k)=0, \quad k \geq \hat{K}\left(x_{0}\right), \quad x_{0} \in {\mathbb{R}}^{+} .
$$
Now, let $V\left(\mathpzc{y}_{0}\right)<x_{0}, \mathpzc{y}(0) \stackrel{\text { a.s. }}{=} \mathpzc{y}_{0} \in \mathbb{R}^{n}$, and it follows from Lemma $3$ that
$$
\mathbb{E}[V(\mathpzc{y}(k))]=0, \quad k \geq \hat{K}\left(x_{0}\right) .
$$
Since $V(\mathpzc{y}(k)), k \in \mathbb{N}$, is a nonnegative random variable, it follows that $V(\mathpzc{y}(k)) \stackrel{\text { a.s. }}{=} 0$ for all $k \geq \hat{K}\left(x_{0}\right)$. Then, it follows from \eqref{22} and \eqref{23} that $\mathpzc{y}(k) \stackrel{\text { a.s. }}{=} 0$ for all $k \geq \hat{K}\left(x_{0}\right)$. Therefore, there exists a stochastic settling-time $\mathbb{E}[K\left(\mathpzc{y}_{0}\right)] \leq$ $\hat{K}\left(x_{0}\right)$ such that $\mathpzc{y}(k)=0, k \geq K\left(\mathpzc{y}_{0}\right)$. Finally, since $\mathbb{E}[K\left(\mathpzc{y}_{0}\right)] \leq \hat{K}\left(x_{0}\right)$, it follows that
\begin{align}
\mathbb{E}\left[K\left(\mathpzc{y}_{0}\right)\right] &\leq \hat{K}\left(x_{0}\right)\nonumber\\&<\lfloor {\alpha^{\frac{1}{1-r_1}}(1-\alpha^{\frac{1}{1-r_1}})}\rfloor + \lfloor \alpha^{-1}(\alpha^{\frac{1}{1-r_2}} -1)\rfloor +3, \nonumber
\end{align}
and hence, Definition $6$ is satisfied. 	\hfill {$\square$}

\section{Example Illustration and Simulation} \label{Sec:3}
This sections provides examples to verify the correctness of the presented fixed-time stability results. Example 1 is presented for deterministic systems without uncertainties and perturbations. Example 2 is a counterexample that shows that if the Lyapunov conditions for a deterministic system guarantees its fixed-time stability, by adding noise to the system, the same Lyapunov conditions only guarantee exponential stability in probability, and not fixed-time stability in probability. This example clearly shows that moving from a fixed-time stable deterministic system to a stochastic system with the same dynamics, one might look for new Lyapunov function candidates than the one used for the deterministic system to show its fixed-time stability in probability, if there exists one. \vspace{6pt}

\textbf{Example 1.} \emph{(Fixed-time stable deterministic discrete-time system)}  
	Consider the scalar nonlinear DT system given as follows
	\begin{align}
			y(k+1)=& ay(k)-\alpha'sign(y(k))\times \nonumber \\&\min\{|y(k)|/\alpha',\max\{|y(k)|^{r'_1}, |y(k)|^{r'_2}\}\}, \label{exm}
		\end{align}
where $y(k) \in \mathbb{R}$ , $k \in \mathbb{N}$, $\frac{1}{2} <a \le 1$, $\alpha' \in (0, 1)$, $r'_1 \in (0, 1)$ and $r'_2 > 1$. Now, using Theorem 1, it is shown that the zero solution $y(k)=0$ to \eqref{exm} with $a=1$ is globally fixed-time stable. Consider $V(y(k))=y^2(k)$ and $y_L<y(0)<y_H$ where $y_L=\alpha'^{\frac{1}{1-r_1}}$ and $y_H=\alpha'^{\frac{1}{1-r_2}}$ (Note that if $y(0)>y_H$ or $y(0)<y_L$, then the zero solution $y(k) = 0$ for \eqref{exm} with $a=1$
is fixed-time stable with $K(y(0)) = 1$).

The difference of $V(y(k))=y^2(k)$ is as follows,
\begin{align}
\Delta &V(y(k))=[ay(k)-\alpha'sign(y(k)) \times \nonumber \\ &\min\{|y(k)|/\alpha',\max\{|y(k)|^{r'_1}, |y(k)|^{r'_2}\}\}]^2-y^2(k)\nonumber \\ =& (ay(k))^2\nonumber \\ &-2a\alpha'|y(k)|\min\{|y(k)|/\alpha',\max\{|y(k)|^{r'_1}, |y(k)|^{r'_2}\}\} \nonumber \\ &+(\alpha'\min\{|y(k)|/\alpha',\max\{|y(k)|^{r'_1}, |y(k)|^{r'_2}\}\})^2-y^2(k) \nonumber \\=& (a^2-1)y^2(k)\nonumber \\ &+ \alpha'\min\{|y(k)|/\alpha',\max\{|y(k)|^{r'_1}, |y(k)|^{r'_2}\}\} \times\nonumber \\& (- 2a|y(k)|+\alpha' \min\{|y(k)|/\alpha',\max\{|y(k)|^{r'_1}, |y(k)|^{r'_2}\}\} ). \label{exmDelV}
\end{align}


Using the fact that 
\begin{align}
|y(k)|>\alpha' \min\{|y(k)|/\alpha',\max\{|y(k)|^{r'_1}, |y(k)|^{r'_2}\}\}, \label{fac}
\end{align}
one has 
\begin{align}
&-2a|y(k)|+\alpha' \min\{|y(k)|/\alpha',\max\{|y(k)|^{r'_1}, |y(k)|^{r'_2}\}\}<\nonumber \\&(1-2a)\alpha' \min\{|y(k)|/\alpha',\max\{|y(k)|^{r'_1}, |y(k)|^{r'_2}\}\}. \label{fac1}
\end{align}
Therefore, using \eqref{fac1}, \eqref{exmDelV} leads to,
\begin{align}
\Delta &V(y(k))\le (a^2-1)y^2(k) \nonumber \\& +(1-2a) {\alpha'}^2 \min\{y^2(k)/\alpha'^2,\max\{y^{2r'_1}(k), y^{2r'_2}(k)\}\}, \label{exmDelV3}
\end{align}
where using $V(y(k))=y^2(k)$ one can rewrite \eqref{exmDelV3} as follows,
\begin{align}
&\Delta V(y(k))\le(a^2-1)V(y(k)) \nonumber \\ &
+(1-2a){\alpha'}^2 \min\{V(k)/\alpha'^2,\max\{V^{r'_1}(k), V^{r'_2}(k)\}\}. \label{exmDelV24}
\end{align}
Sice $\frac{1}{2} < a \le 1$, \eqref{exmDelV24} is rewritten as
\begin{align}
&\Delta V(y(k))\le 
-\beta{\alpha'}^2 \min\{V(k)/\alpha'^2,\max\{V^{r'_1}(k), V^{r'_2}(k)\}\}, \label{exmDelV34}
\end{align}
where $\beta=(2a-1)$ and for $\frac{1}{2} < a \le 1$, $0<\beta \le 1 $.

For $ a = 1 $, \eqref{exmDelV34} leads to
\begin{align}
&\Delta V(y(k))\le 
-{\alpha'}^2 \min\{V(k)/\alpha'^2,\max\{V^{r'_1}(k), V^{r'_2}(k)\}\}. \label{exmDelV35}
\end{align}
which is analogous to \eqref{Lyap} where $\alpha={\alpha'}^2$, $r_1=r'_1$ and $r_2=r'_2$, and all the parameters conditions mentioned in Theorem 1 are satisfied. Therefore, it is shown that system \eqref{exm} with $ a = 1 $ is globally fixed-time stable. Based on \eqref{FTB1}, the fixed upper bound for the settling-time function of system \eqref{exm} with $ a = 1 $ is 
\begin{align}
			K^*= \lfloor {{\alpha'}^{\frac{2}{1-r'_1}} (1-\alpha'^{\frac{2}{1-r_1'}})}\rfloor + \lfloor \alpha'^{-2}({\alpha'}^{\frac{2}{1-r'_2}} -1)\rfloor +3.\label{exmFTB}
\end{align}

The state trajectory of the system \eqref{exm} with $ a = 1 $ is simulated in Fig. 1 for 4 different values of parameters $\alpha'$, $r'_1$ and $r'_2$ to verify the fixed-time convergence of the system \eqref{exm} with $ a = 1 $ where in Cases 1-4, $y(0)=20$ such that $y_L<y(0)<y_H$ in Cases 1-3 and $y(0)>y_H$ for Case 4. As depicted in Fig. 1, the settling-time is less than $K^*$ for Cases 1-3 where $K^*$ is calculated using \eqref{exmFTB} and given in Table 1, and as mentioned in \eqref{bdB1}, for case 4, $K(y(0))=1$. 
\begin{table}
\begin{center}
		\caption{Parameters $\alpha' $, $r'_1$, $r'_2 $, and fixed-time upper bound of settling-time function ($K^*$) for \eqref{exm} with $a=1$ and the initial condition $y(0)=8$.} 
\begin{tabular}{ | c| c | c| c | c | c | c |  } 
\hline
{ }  &  {$\alpha'$}&  {$r_1'$}&  {$r_2'$}&  {$K^*$}&  {$y_L$}&  {$y_H$} \\
\hline
{Case 1}  & {0.4} & {0.2}& {1.2}&{ $59601$ }& {0.31}&{ $97.6$ } \\
\hline
{Case 2}  & {0.7} & {0.9} & {1.1}& {2558} & {0.02}& {35.4}\\
\hline
{Case 3} &  {0.3} & {0.6} & {1.3} & {$34002$}& {0.04} & {$55.3$}  \\ 
\hline
{Case 4} & {0.7} & {0.9} & {10} & {$3$}& {0.02} & {$1.04$}  \\ 
 \hline
\end{tabular}
\end{center}
\end{table}

\begin{figure}[thpb]
	\centering
	\includegraphics[height=2.8 in, width=3.5 in]{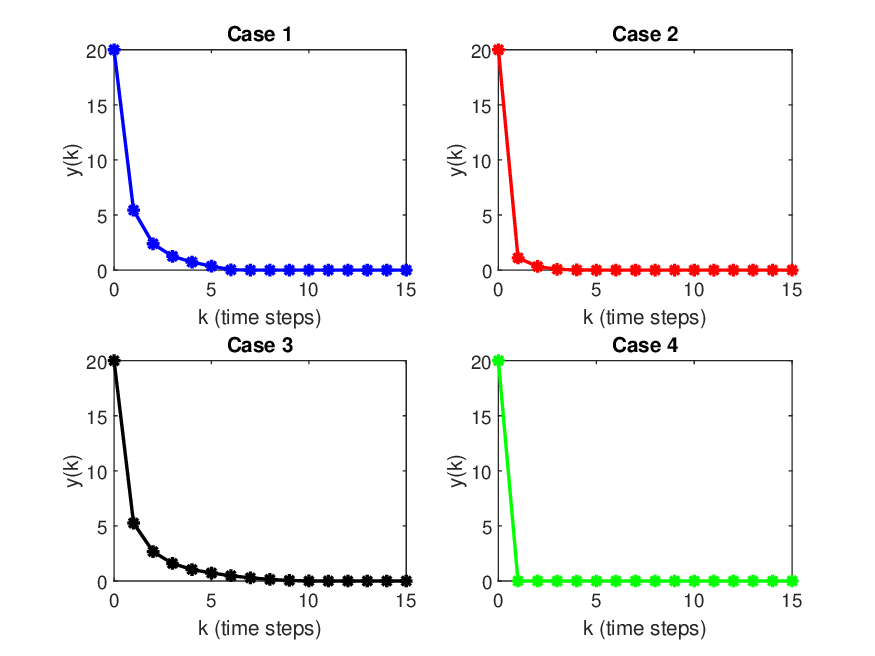}
	\caption{Different fixed times of convergence for system \eqref{exm} with $a=1$ and different values of $\alpha' $, $r'_1$ and $r'_2 $.}
	\label{figurelabel}
\end{figure}

\begin{figure}[thpb]
	\centering
	\includegraphics[height=2.8 in, width=3.5 in]{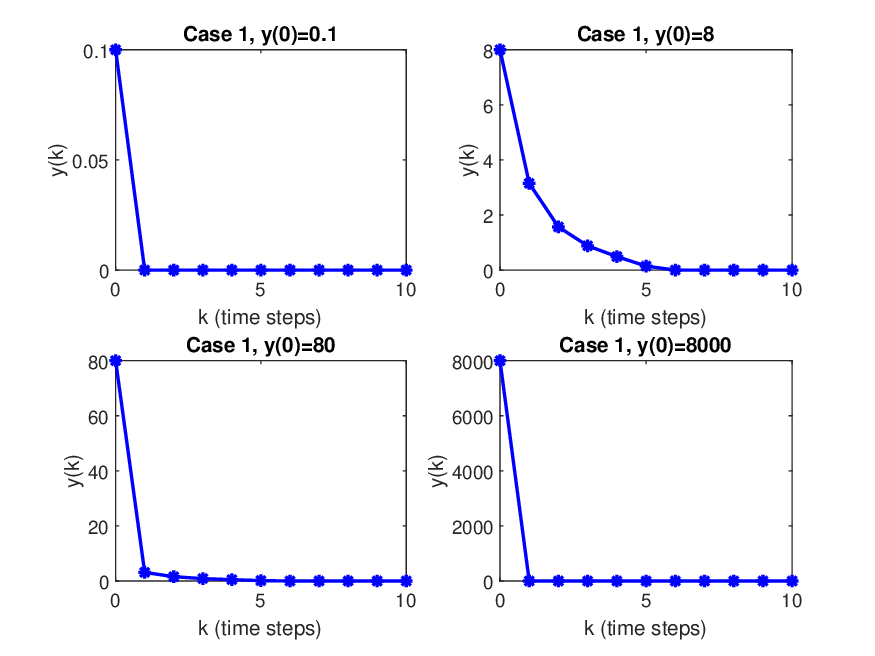}
	\caption{Fixed-time convergence for Case 1 of system \eqref{exm} with $a=1$ for different initial values.}
	\label{figurelabel}
\end{figure}

In Fig. 2, the state trajectory of system \eqref{exm} with $ a = 1 $ and Case 1 parameters ($\alpha'=0.4$, $r_1'=0.2$, $r_2'=1.2$) is simulated for 4 different initial conditions, $y(0)=0.1, (y(0)<y_L)$, $y(0)=8, (y_L<y(0)<y_H)$, $y(0)=80, (y_L<y(0)<y_H)$ and $y(0)=8000, (y_H<y(0))$ where as expected for $y(0)=0.1$ and $y(0)=8000$, the settling-time is $K(y(0))=1$, and for $y(0)=8$ and $y(0)=80$ the convergence to zero is achieved in few steps which ensures $K(y(0)) \le K^*$. 

However, for $\frac{1}{2} < a < 1$, based on \eqref{exmDelV34}, Lemma 2 and a similar procedure to Theorem 1 proof, one can show that the system \eqref{exm} with  $\frac{1}{2} < a < 1$ is exponentially stable.
\vspace{6pt}

\textbf{Example 2.} \emph{(Lyapunov function candidate: from deterministic fixed-time stable systems to their stochastic counterparts)} 

In this counterexample we show that the deterministic global fixed-time stable system may not preserve its fixed-time stability under the same Lyapunov function candidate after it is exposed to stochastic noise.

	Consider the scalar stochastic nonlinear DT system as follows
	\begin{align}
			\mathpzc{y}&(k+1)=  a \mathpzc{y}(k)-\alpha'sign(\mathpzc{y}(k))\times \nonumber \\&\min\{|\mathpzc{y}(k)|/\alpha',\max\{|\mathpzc{y}(k)|^{r'_1}, |\mathpzc{y}(k)|^{r'_2}\}\} + b \mathpzc{y}(k) \nu(k), \label{exm2}
		\end{align}
where $\mathpzc{y}(k) \in \mathbb{R}$ , $k \in \mathbb{N}$, $\alpha' \in (0, 1)$, $r'_1 \in (0, 1)$ and $r'_2 > 1$, $\nu(k) \in \mathbb{R}$ is a zero-mean stochastic noise with $\mathbb{E}[\nu(k)]=0$ and $\mathbb{E}[\nu^2(k)]=\sigma^2$, $\frac{1}{2} <a \le 1$ and $b < \sqrt{\frac{1-a^2}{\sigma^2}}$.

Now, using Theorem 3 and the results of Example 1, it is shown that the zero solution $\mathpzc{y}(k) \stackrel{\text { a.s. }}{=} 0$ to \eqref{exm2} (the stochastic version of \eqref{exm}) does not show global fixed-time stability in probability for $a=1$ but preserves its exponential stability in probability for $\frac{1}{2} <a < 1$, using the same Lyapunov function as in Example 1. 

Consider $V(\mathpzc{y}(k))=\mathpzc{y}^2(k)$ such that for \eqref{exm2}, one has
\begin{align}
\boldsymbol{\Delta} &V(\mathpzc{y}(k))=\nonumber \\ &\mathbb{E}[(a\mathpzc{y}(k)-\alpha'sign(\mathpzc{y}(k))\min\{|\mathpzc{y}(k)|/\alpha',\nonumber \\ &\max\{|\mathpzc{y}(k)|^{r'_1}, |\mathpzc{y}(k)|^{r'_2}\}\}+  b \mathpzc{y}(k) \nu(k))^2]-\mathpzc{y}^2(k)\nonumber \\ =& \mathbb{E}[a^2\mathpzc{y}^2(k)+(\alpha'\min\{|\mathpzc{y}(k)|/\alpha',\max\{|\mathpzc{y}(k)|^{r'_1}, |\mathpzc{y}(k)|^{r'_2}\}\})^2 \nonumber \\&+b^2\mathpzc{y}^2(k)\nu^2(k) +2ab\mathpzc{y}^2(k)\nu(k)
\nonumber \\ &-2a\alpha'|\mathpzc{y}(k)|\min\{|\mathpzc{y}(k)|/\alpha',\max\{|\mathpzc{y}(k)|^{r'_1}, |\mathpzc{y}(k)|^{r'_2}\}\} \nonumber \\ &-2b\alpha'\nu(k)|\mathpzc{y}(k)|\min\{|\mathpzc{y}(k)|/\alpha',\max\{|\mathpzc{y}(k)|^{r'_1}, |\mathpzc{y}(k)|^{r'_2}\}\}]\nonumber \\ &-\mathpzc{y}^2(k) =
(a^2+b^2\sigma^2-1) \mathpzc{y}^2(k)\nonumber \\ &+\alpha'^2\min\{|\mathpzc{y}(k)|^2/{\alpha'^2},\max\{|\mathpzc{y}(k)|^{2r'_1}, |\mathpzc{y}(k)|^{2r'_2}\}\} \nonumber \\ &-2a\alpha'|\mathpzc{y}(k)|\min\{|\mathpzc{y}(k)|/\alpha',\max\{|\mathpzc{y}(k)|^{r'_1}, |\mathpzc{y}(k)|^{r'_2}\}\}
. \label{exmDelV2}
\end{align}

Using \eqref{fac}, \eqref{exmDelV2} leads to,
\begin{align}
&\boldsymbol{\Delta} V(\mathpzc{y}(k))\le (a^2+b^2\sigma^2-1) \mathpzc{y}^2(k) \nonumber \\ &- (2a-1){\alpha'}^2 \min\{\mathpzc{y}^2(k)/\alpha'^2,\max\{\mathpzc{y}^{2r'_1}(k), \mathpzc{y}^{2r'_2}(k)\}\}, \label{exmDelV13}
\end{align}
where using $V(y(k))=y^2(k)$ one can rewrite \eqref{exmDelV13} as follows,
\begin{align}
&\boldsymbol{\Delta}V(\mathpzc{y}(k))\le (a^2+b^2\sigma^2-1) V(k) \nonumber \\ &- \beta{\alpha} \min\{V(k)/\alpha,\max\{V(k)^{r_1}(k), V(k)^{r_2}(k)\}\}, \label{exmDelV133}
\end{align}
where $\beta=2a-1$, $\alpha={\alpha'}^2$, $r_1=r'_1$ and $r_2=r'_2$.

For $a=1$, \eqref{exmDelV133} reduces to 
\begin{align}
&\boldsymbol{\Delta} V(\mathpzc{y}(k))\le b^2\sigma^2 V(k) \nonumber \\ &- {\alpha} \min\{V(k)/\alpha,\max\{V(k)^{r_1}(k), V(k)^{r_2}(k)\}\}. \label{exmDelV333}
\end{align}
However, \eqref{exmDelV333} can not support the global fixed-time stability in probability of the system \eqref{exm2} with $a=1$, due to the injected noise stochasticity, while in Example 1 it was shown that the same system without noise is fixed-time stable. 

For $\frac{1}{2} < a < 1$ and $b < \sqrt{\frac{1-a^2}{\sigma^2}}$, one has $0<\beta<1$ and $a^2+b^2\sigma^2-1<0$. Thus, using \eqref{exmDelV133} one obtains 
\begin{align}
&\boldsymbol{\Delta} V(\mathpzc{y}(k))\le \nonumber \\ &- \beta{\alpha} \min\{V(k)/\alpha,\max\{V(k)^{r_1}(k), V(k)^{r_2}(k)\}\}. \label{exmDelV533}
\end{align}
By using \eqref{exmDelV533}, Lemma 2 and a similar procedure to Theorem 3 proof, one can show that the system \eqref{exm2} with  $\frac{1}{2} < a < 1$ and $b < \sqrt{\frac{1-a^2}{\sigma^2}}$ is exponentially stable in probability. Therefore, the stochastic system \eqref{exm2} preserves exponential stability in probability for $\frac{1}{2} < a < 1$ and $b < \sqrt{\frac{1-a^2}{\sigma^2}}$.

\section{Conclusion}
This paper addressed the fixed-time stability for deterministic and stochastic discrete-time (DT) autonomous systems based on fixed-time Lyapunov stability analysis. Novel Lyapunov conditions are derived under which the fixed-time stability of autonomous DT deterministic and stochastic systems is certified. The sensitivity to perturbations for fixed-time stable DT systems is analyzed and the analysis shows that fixed-time attractiveness can be resulted from the presented Lyapunov conditions. For both cases of fixed-time stable and fixed-time attractive systems, the fixed upper bounds of the settling-time functions are given. For future work, we intend to employ the introduced DT systems fixed-time stability analysis for control and identification of such systems.

\bibliographystyle{IEEEtran}
\bibliography{IEEEabrv,ref.bib}
 
\end{document}